\documentclass[12pt,a4paper]{amsart}

\usepackage{amsmath,amsfonts,amssymb}

\usepackage[hmargin=2cm,vmargin=2cm]{geometry}

\usepackage{hyperref}
\usepackage{nameref,zref-xr}                    

\usepackage{amsthm,amscd}
\usepackage{epsfig}
\usepackage{color}
\usepackage[all]{xy}
\usepackage{tikz}
\usepackage{graphics, setspace}
\usepackage{breqn}
\usepackage{amstext}
\usepackage{array}

\newcommand{\p}{{\partial}}

\def\d{{\partial}}

\newcommand{\ZZ}{\mathbb{Z}}

\newcommand{\bt}{{\bf t}}

\newtheorem{theorem}{Theorem}[section]
\newtheorem{proposition}[theorem]{Proposition}

\usepackage{color}

\def\&{\vspace{-5pt}&}

\numberwithin{equation}{section}

\begin{document}

\title{Integrable systems associated to open extensions of type A and D Dubrovin--Frobenius manifolds}

\author{Alexey Basalaev}
\address{A. Basalaev:\newline Faculty of Mathematics, National Research University Higher School of Economics, Usacheva str., 6, 119048 Moscow, Russian Federation, and \newline
Skolkovo Institute of Science and Technology, Nobelya str., 3, 121205 Moscow, Russian Federation}
\email{a.basalaev@skoltech.ru}

\date{\today}

 \begin{abstract}
    We investigate the solutions to open WDVV equation, associated to type A and D Dubrovin--Frobenius manifolds. We show that these solutions satisfy some stabilization condition and associate to both of them the systems of commuting PDEs. In the type A we show that the system of PDEs constructed coincides with the dispersionless modifiled KP hierarchy written in the Fay form.
 \end{abstract}
 \maketitle
\section{Introduction}

Given a Dubrovin--Frobenius manifold, there are several constructions of an integrable system associated to it (cf. \cite{DZ}, \cite{GM} and \cite{B2}). 
In particular, it was proved that for type $A_N$ and $D_N$ Dubrovin--Frobenius manifolds all these constructions provide the corresponding Drinfeld--Sokolov hierarchies (cf. \cite{LRZ}). The $A_N$ Drinfeld--Sokolov hierarchy can be defined as the reductions of the KP and type $D_N$ Drinfeld--Sokolov hierarchy can be defined as a 1--component reduction of the 2--component BKP hierarchies.

In \cite{BDbN} the authors proposed a new construction of a system of commuting PDEs associated to the family of $A$--type and $D$--type Dubrovin--Frobenius manifolds. 
One of the main properties of this construction was the following: consistency of the system of PDEs constructed was derived from WDVV equation on the potential of a Dubrovin--Frobenius manifold.
Comparing to the previous approaches this construction associated dispersionless KP hierarchy and dispersionless 1-component reduced 2--component BKP hierarchy in $A$ and $D$ types respectively. 

\subsection{Open WDVV equation}
Motivated by studies of the open Gromov--Witten theories the new system of PDEs called \textit{open WDVV} was introduced in \cite{HS12}. Given a solution $F^c = F^c(t_1,\dots,t_N)$ to WDVV equation with the metric $\eta$, open WDVV equation on a function $F^o = F^o(t_0,t_1,\dots,t_N)$ reads
\begin{align}
\label{eq:open WDVV}
\frac{\d^3F^c}{\d t_\alpha\d t_\beta\d t_\mu}\eta^{\mu\nu}\frac{\d^2F^o}{\d t_\nu\d t_\gamma}+\frac{\d^2F^o}{\d t_\alpha\d t_\beta}\frac{\d^2F^o}{\d t_0\d t_\gamma}=\frac{\d^3F^c}{\d t_\gamma\d t_\beta\d t_\mu}\eta^{\mu\nu}\frac{\d^2F^o}{\d t_\nu\d t_\alpha}+\frac{\d^2F^o}{\d t_\gamma\d t_\beta}\frac{\d^2F^o}{\d t_0\d t_\alpha},
\end{align}
for any fixed $0 \le \alpha,\beta,\gamma \le N$. 

Similarly to ``classical'' WDVV equation, open WDVV equation is associativity equation of the product defined by
\[
    \frac{\p}{\p t_\alpha} \circ \frac{\p}{\p t_\beta} = \sum_{\gamma,\delta=1}^{N} \frac{\p^3 F^c}{\p t_\alpha \p t_\beta \p t_\gamma} \eta^{\gamma\delta} \frac{\p}{\p t_\delta} + \frac{\p^2 F^o}{\p t_\alpha \p t_\beta} \frac{\p}{\p t_0}, \quad 0 \le \alpha,\beta \le N.
\]
From this point of view the function $F^o$, solving open WDVV equation, defined an \textit{open extensions} of a Dubrovin--Frobenius manifold given by $F^c$.

For $F^c$ being Dubrovin--Frobenius potential of type $A_N$ or $D_N$ the solutions to open WDVV equation were investigated in \cite{BB2}. In particular, for the $A_N$ and $D_N$ Dubrovin--Frobenius potentials $F^c_{A_N}$ and $F^c_{D_N}$ the authors constructed the functions $F^o_{A_N}$ and $F^o_{D_N}$, being solutions to open WDVV equation. The function $F^o_{A_N}$ appeared to be polynomial in $t_0,t_1,\dots,t_N$ and $F^o_{D_N}$ polynomial in $t_1,\dots,t_N$, but Laurent polynomial in $t_0$.

The integrable hierarchies associated to the solutions of open WDVV equations were investigated in \cite{BCT1, BCT2} and \cite{A}. In all these references only $A_N$ type was assumed. 

\subsection{System of commuting PDEs}
Extending the approach of \cite{BDbN} we associate to the families $\lbrace (F^c_{A_N},F^o_{A_N}) \rbrace_{N \ge 1}$ and $\lbrace(F^c_{D_N},F^o_{D_N})\rbrace_{N\ge 4}$ the systems of commuting PDEs. Like in loc.cit. the first step on this way is the following stabilization result.

In the A case for any $N < M$ and $ 1\le \alpha \le N$ we have
    \[
      \frac{\p F_{A_N}^o}{\p t_\alpha} \mid_{t_{N+1-\beta} = s_\beta} \ = \frac{\p F_{A_M}^o}{\p t_\alpha} \mid_{t_{M+1-\beta} = s_\beta},
    \]
    assumed as polynomials in $t_0,s_1,\dots,s_{M+1}$. In the D case for any $4 \le N < M$ and $ 1\le \alpha < N$ we have
    \begin{align*}
      \frac{\p F_{D_N}^o}{\p t_\alpha} \mid_{t_{N-\beta} = s_\beta, \ t_{N} = \bar s_{1}} \ = \frac{\p F_{D_M}^o}{\p t_\alpha} \mid_{t_{M-\beta} = s_\beta, \ t_{M} = \bar s_{1}},
      \\
      \frac{\p F_{D_N}^o}{\p t_{N}} \mid_{t_{N-\beta} = s_\beta, \ t_{N} = \bar s_{1}} \ = \frac{\p F_{D_M}^o}{\p t_{M}} \mid_{t_{M-\beta} = s_\beta, \ t_{M} = \bar s_{1}},
    \end{align*}
    assumed as polynomials in $\bar s_1,s_1,\dots,s_{M+1}$ and Laurent polynomials in $t_0$.

Both systems of equations that we write express the higher derivatives $\p_\alpha\p_\beta f$ of a solution $f$ via $\p_1\p_0 f, \p_1\p_1 f, \p_1\p_2 f,\dots$ that may be viewed as ``initial data''.

For $A$ type this is the system of PDEs on a function $f = f(t_0,t_1,t_2,\dots)$
\begin{align*}
    \p_\alpha\p_\beta f &= \frac{\p^2 F^c_{A_\kappa}}{\p t_\alpha \p t_\beta} \mid_{ t_{\kappa+1-\gamma} = \p_1 \p_\gamma f}, \qquad \kappa = \alpha+\beta+1.
    \\
    \p_0 \p_\alpha f &= \frac{\p F^o_{A_\kappa}}{\p t_\alpha} \mid_{ t_{\kappa+1-\gamma} = \p_1 \p_\gamma f, \ t_0 = \p_1 \p_0 f}
\end{align*}
where the abbreviate $\p_\alpha = \p / \p t_\alpha$. We show in Section~\ref{section: A} that this system is well--defined and consistent. We also prove in Theorem~\ref{theorem: dispersionless mKP} that this system of equations coincides with the dispersionless modified KP hierarchy written in the Fay form. We also provide the construction extending new type A system to the full mKP hierarchy.

For D type we introduce the system of PDEs on $f = f(\bar t_1, t_0,t_1,\dots)$
\begin{align*}
    \p_\alpha\p_\beta f &= \frac{\p^2 F^c_{D_\kappa}}{\p t_\alpha \p t_\beta} \mid_{ t_{\kappa-\gamma} = \p_1 \p_\gamma f, \ t_\kappa = \p_1 \bar\p_1 f}, \qquad \kappa = \alpha+\beta-1.
    \\
    \p_\alpha\bar\p_1 f &= \frac{\p^2 F^c_{D_\kappa}}{\p t_\alpha \p t_\kappa} \mid_{ t_{\kappa-\gamma} = \p_1 \p_\gamma f, \ t_\kappa = \p_1 \bar\p_1 f},
    \\
    \p_0 \p_\alpha f &= \frac{\p F^o_{D_\kappa}}{\p t_\alpha} \mid_{ t_{\kappa-\gamma} = \p_1 \p_\gamma f, \ t_\kappa = \p_1 \bar\p_1 f, \ t_0 = \p_1 \p_0 f}, \qquad \kappa = \alpha+1,
    \\
    \p_0 \bar \p_1 f \p_0 \p_1 f &= \p_1\bar\p_1 f,
\end{align*}
where we abbreviate $\bar\p_1 = \p/ \p \bar t_1$. It is proved in Theorem~\ref{theorem: D type} that this system of equations is well-defined and consistent.

Like in \cite{BDbN} we derive consistency of both A and D type PDEs from open WDVV equation above. Comparing to the previously cited references our work is the first example of the system of PDEs constructed by the open WDVV solutions.

Following the relation of $D_N$ type Drinfeld--Sokolov hierarchy and also D type hierarchy of \cite{BDbN} to the 2--component BKP hierarchy, one would expect that our D type system is connected to the modified BKP hierarchy. However the modified version of BKP hierarchy is not yet fully settled.

\noindent{\bf Acknowledgements.}
The author acknowledges partial support by International Laboratory of Cluster Geometry HSE University (RF Government grant, agreement no. 075-15-2021-608 from 08.06.2021).

\section{Open potentials of type A and D}

The Dubrovin--Frobenius potentials $F_{A_N}^c$ and $F^c_{D_N}$ were first derived by B.Dubrovin via the geometry of the respective Coxeter groups (cf. \cite{D1}). Since then these potentials were constructed in the several different ways by the other authors. We neither give the construction of these Dubrovin--Frobenius manifolds nor provide the expansion of the respective potentials referencing the reader to \cite{BDbN}. 

\subsection{$A_N$ open potential}
The genus zero $A_N$ open potential $F^o_{A_N}$ was first found in \cite{BCT2}. It is a polynomial in $t_1,\dots,t_N$ and $t_0$ defined by
\begin{equation*}
\left.\frac{\d^{m+k} F^o_{A_N}}{\d t_{\alpha_1}\ldots\d t_{\alpha_m}(\d t_0)^k}\right|_{t_\ast=0}
=
\begin{cases}
(m+k-2)!,&\text{if $\sum_{i=1}^m(N+2-\alpha_i)+k=N+2$},\\
0,&\text{otherwise}.
\end{cases}
\end{equation*}
It is also connected to the $A_N$ unfolding coordinates by the following identity (see \cite{B3})
\begin{equation}\label{eq: open potential derivative}
    \frac{\p F^o_{A_N}}{\p t_0} = \frac{t_0^{N+1}}{N+1} + \sum_{k=1}^N t_0^{k-1} v_k^{\mathrm{A}}(t_1,\dots,t_N)
\end{equation}
where $v_\bullet^{\mathrm{A}}$ are defined by
\begin{align}\label{eq: An essential coordinate via flat}
     v_\gamma^{\mathrm{A}} &=  \sum_{\substack{\alpha_1,\dots,\alpha_N \ge 0 \\ \sum_{k=1}^{N} (N+2-k) \alpha_k = N+2 - \gamma}} \frac{(|\alpha|+\gamma-2)!}{(\gamma-1)!} \prod _{k=1}^{N} \frac{t_k^{\alpha _k}}{\alpha _k!}
\end{align}
for $|\alpha| = \sum_{k=1}^{N} \alpha_k$.

It follows immediately from the definition above that $F^o_{A_N}$ satisfies the following quasihomogeneity condition
\begin{equation}\label{eq: A open potential quasihomogeneity}
    \left( \sum_{k=1}^N (N+2-k) t_k \frac{\p}{\p t_k} + t_0 \frac{\p }{\p t_0}\right) F^o_{A_N} = (N+2) F^o_{A_N}. 
\end{equation}

It was shown in \cite{BB2} that $F^o_{A_N}$ is the only polynomial satisfying this quasihomogeneity condition, s.t. $(F^c_{A_N},F^o_{A_N})$ is a solution to open WDVV equation and $\p_1\p_0 F^o_{A_N} = 1, \p_1\p_\alpha F^o_{A_N} = 0$ for $1 \le \alpha \le N$.

\subsection{$D_N$ open potential}
The genus zero $D_N$ open potential $F^o_{D_N}$ is a polynomial in $t_1,\dots,t_N$ and Laurent polynomial in $t_0$ defined by (cf. \cite[Section 5.2]{BB2})
\begin{gather*}
F^o_{D_N}:=\frac{t_0^{2 N-1}}{2^{N-2}(2 N-1) (2N-2)} + \frac{t_N^2}{2 t_0} + 
\sum_{k=1}^{N-1} \frac{v_k^{\mathrm{D}} t_0^{2 k-1}}{2^{k-1}(2 k-1)}
\end{gather*}
with the functions $v_\bullet^{\mathrm{D}}$ given by (see \cite[Corollary 4.6]{BDbN})
\begin{align}\label{eq: Dn essential coordinate via flat}
    v_b^{\mathrm{D}} &=  \sum_{\substack{\alpha_1,\dots,\alpha_{N-1} \ge 0 \\ \sum_{k=1}^{N-1} (N-k) \alpha_k = N - b}} \frac{(|\alpha|+2b-3)!}{(2b-2)!} \prod _{k=1}^{N-1} \frac{t_k^{\alpha _k}}{\alpha _k!}.
\end{align}

The Laurent polynomial $F^o_{D_N}$ satisfies the following quasihomogeneity condition
\begin{equation}\label{eq: D open potential quasihomogeneity}
    \left( \sum_{k=1}^{N-1} 2(N-k) t_k \frac{\p}{\p t_k} + N t_N \frac{\p}{\p t_N} + t_0 \frac{\p }{\p t_0}\right) F^o_{D_N} = (2N-1) F^o_{D_N}. 
\end{equation}

It is immediate to see that the variable $t_N$ plays a special role in $F^o_{D_N}$. In particular, the functions $v_b^{\mathrm{D}}$ do not depend on $t_N$ and the only non--polynomial summand of $F^o_{D_N}$ is at the same time the only appearance of $t_N$ in the open potential.
This variable is also special for $F^c_{D_N}$. In particular, we have (see \cite[Section 3.3]{BDbN})
\begin{equation}\label{eq: DN^c expression}
    F^c_{D_N} = \frac{1}{2}t_1^2t_{N-1} +\frac{1}{2}t_1 t_{N}^2 + \frac{1}{2} t_1 \sum_{\alpha,\beta=2}^{N-1} t_\alpha t_{N-\alpha} + \phi(t_2,\dots,t_{N-1}) + v_1^{\mathrm{D}} \cdot \frac{t_N^2}{2},
\end{equation}
for some polynomial $\phi$ that does not depend on $t_N$ again. This speciality of $t_N$ will result in the special form of the PDEs that we obtain in the D case.

\section{Dispersionless open--closed system of type A}\label{section: A}
The following stabilization condition was proved in \cite[Theorem 4.1]{BDbN}.
For any $M > N \ge 1$ and $\alpha,\beta$, s.t. $1 \le \alpha,\beta \le N$, $\alpha + \beta \le N+1$ we have
\[
    \left.\frac{\p^2 F^c_{A_{N}}}{\p t_\alpha \p t_\beta} \right|_{\forall \gamma\; t_{N+1-\gamma} = s_\gamma} = \left.\frac{\p^2 F^c_{A_M}}{\p t_\alpha \p t_\beta} \right|_{\forall \gamma\; t_{M + 1- \gamma} = s_\gamma},
\]
understood as an equality of polynomials in $s_\bullet$.
We show that it also holds for $F^o_{A_N}$.

\begin{proposition}
    For any $N < M$ and $ 1\le \alpha \le N$ we have
    \[
      \frac{\p F_{A_N}^o}{\p t_\alpha} \mid_{t_{N+1-\beta} = s_\beta} = \frac{\p F_{A_M}^o}{\p t_\alpha} \mid_{t_{M+1-\beta} = s_\beta},
    \]
    assumed as polynomials in $t_0,s_1,\dots,s_{M+1}$.
\end{proposition}
\begin{proof}
    Note that the change of the variables above does not affect the variable $t_0$ and we can consider $F^o_{A_N}$ and $F^o_{A_M}$ as polynomials in $t_0$.

    For the summands of $F^o_{A_N}$ and $F^o_{A_M}$ involving the variable $t_0$ the statement follows from \eqref{eq: open potential derivative} and Lemma 4.2 of \cite{BDbN}.

    Consider the free terms in $t_0$. We have for any $\kappa \ge 1$
    \[
    \left.\frac{\d^{m+1} F^o_{A_\kappa}}{\d t_{\alpha} t_{\kappa+1-\beta_1}\ldots\d t_{\kappa+1-\beta_m}}\right|_{t_\ast=0}
    =
    \begin{cases}
    (m-1)! &\text{ if }\quad m + \sum_{i=1}^m \beta_i = \alpha,\\
    0,&\text{otherwise}.
    \end{cases}
    \]
    Both the value and the condition above do not depend on $\kappa$ what approves the statement.
\end{proof}

For any $\alpha,\beta \in \ZZ_{\ge 0}$ and $\gamma_\bullet \in \ZZ_{\ge 0}$ take $N = \alpha+\beta+1$ and denote $\bar x := N+1-x$. Set
\begin{align}\label{eq:RFder}
    R^{\mathrm{A}}_{\alpha,\beta; \gamma_1,\dots,\gamma_m} &= 
        \frac{1}{m!} \left.\frac{\p^{m+2} F_{A_N}^c}{\p t_\alpha\p t_\beta \p t_{\bar\gamma_1}\cdots \p t_{\bar\gamma_m}} \right|_{\bt = 0},
    \\
    R^{\mathrm{A},ext}_{\alpha; \gamma_1,\dots,\gamma_m} &= 
        \frac{1}{m!} \left.\frac{\p^{m+1} F_{A_N}^o}{\p t_\alpha \p t_{\bar\gamma_1}\cdots \p t_{\bar\gamma_m}} \right|_{\bt =  0}.
\end{align}
The proposition above and \cite[Theorem 4.1]{BDbN} assure that these coefficients are well--defined for all natural $\alpha,\beta,\gamma_\bullet$.
Note that $R^{\mathrm{A}}_{0,\beta; \gamma_1,\dots,\gamma_m} = 0$ by construction.

Consider the system of PDEs
\begin{align}
    \p_\alpha\p_\beta f & = \sum_{m \ge 1} \sum_{\gamma_\bullet} R^{\mathrm{A}}_{\alpha,\beta; \gamma_1,\dots,\gamma_m} \prod_{k=1}^m \p_1\p_{\gamma_k} f,
    \label{eq: A KP system}
    \\
    \p_{0}\p_\alpha f & = \sum_{m \ge 1} \sum_{\gamma_\bullet} R^{\mathrm{A},ext}_{\alpha; \gamma_1,\dots,\gamma_m} \prod_{k=1}^m \p_1\p_{\gamma_k} f.
    \label{eq: A extended system}
\end{align}

One notes immediately that these PDEs coincide with the PDEs presented in Introduction.

\begin{proposition}
    The system \eqref{eq: A KP system},\eqref{eq: A extended system} is compatible. The function $\tilde f = F_{A_N}^c + \int F^o_{A_N} dt_0$ satisfies this system for $\alpha+\beta \le N+1$.
\end{proposition}
\begin{proof}
    Compatibility of the system written follows from open WDVV equation in the same way it was proved in \cite[Proposition 2.1]{BDbN}. The proof is parallel to the one given in D case in the proof of Theorem~\ref{theorem: D type}.
    
    Because of the special dependance of $F^c_{A_N}$ and $F^o_{A_N}$ on the variable $t_1$ we have
    \[
        \p_1\p_\gamma \tilde f = \begin{cases}
                           \p_1\p_\gamma F^c_{A_N} = t_{\overline \gamma}, \quad &\text{ if } \quad \gamma > 0,
                           \\
                           \p_1 F^o_{A_N} = t_0 , \quad &\text{ if } \quad \gamma = 0.
                          \end{cases}
    \]
    Let's show that applied to $\tilde f$ assumed both \eqref{eq: A KP system} and \eqref{eq: A extended system} just provide the series expansions of $\p_\alpha\p_\beta F^c_{A_N}$ and $\p_\alpha F^o_{A_N}$.
    
    It follows from \eqref{eq: A open potential quasihomogeneity} that $\p_\alpha \p_\beta F^o_{A_N} \equiv 0$ whenever $\alpha+\beta \le N+1$.
    The function $F^c_{A_N}$ does not depend on $t_0$ and therefore \eqref{eq: A KP system} holds for $\tilde f$.
    
    Similarly $\p_0\p_\alpha \tilde f = \p_\alpha F^o_{A_N}$ what approves \eqref{eq: A extended system} for $\tilde f$.    
\end{proof}

\subsection{The flows}
For any $f = f(t_0,t_1,t_2,\dots)$ denote $p_k := \p_1\p_k f$. 
The first flows read
\begin{align*}
    \p_0\p_1 f &= p_0,
    \\
    \p_0\p_2 f &= \frac{p_0^2}{2}+p_1,
    \\
    \p_0\p_3 f &= \frac{p_0^3}{3}+p_1 p_0+p_2,
    \\
    \p_0\p_4 f &= \frac{p_0^4}{4}+p_1 p_0^2+p_2 p_0+\frac{p_1^2}{2}+p_3.
\end{align*}

Denote by $P_{ij}(\gamma_1,\dots,\gamma_m)$ the number of all partitions of $i_1,\dots,i_m$ of $i$ and $j_1,\dots,j_m$ of $j$, s.t. $i_k + j_k = \gamma_k+1$ for all $k$. 
It was computed in \cite[Corollary 5.1]{BDbN} that the flows of \eqref{eq: A KP system} read:
\begin{equation}\label{eq: dKP flows expansion}
    \p_i\p_j f = \sum_{m \ge 1} \frac{(-1)^{m-1}}{m} \sum_{\gamma_1+\dots+\gamma_m = i+j-m} P_{ij}(\gamma_1,\dots,\gamma_m) \prod_{k=1}^m \p_1\p_{\gamma_k} f.
\end{equation}

\begin{proposition}\label{prop: A ext expanded}
    Equation \eqref{eq: A extended system} is equivalent to the following equality of the formal power series in $z$
    \[
        \sum_{\alpha \ge 1} \p_0\p_\alpha f \cdot z^\alpha = - \log \left[1 - \sum_{\alpha \ge 1} \p_1\p_\alpha f \cdot z^{\alpha+1} - z \cdot \p_0\p_1 f \right].
    \]
\end{proposition}
\begin{proof}
 Introduce the notation $x_\alpha := \p_1\p_{\alpha-1}f$ and $y_\alpha := \p_0\p_\alpha f$. Let also $\phi_{k,m} := \sum x_{\gamma_1}\dots x_{\gamma_m}$ where the summation is taken over all $\gamma_i \ge 2$, s.t. $\gamma_1+\dots+\gamma_m=k$. 
Then \eqref{eq: A extended system} can be written using \eqref{eq: open potential derivative} as
\[
    y_\alpha = \sum_{m=1}^\alpha \frac{\phi_{\alpha,m}}{m} + \sum_{m=0}^{\alpha-1} \sum_{k=m+1}^\alpha \frac{y_1^{\alpha+1-k}}{(\alpha+1-k)!} \frac{(m+\alpha-k)!}{m!} \phi_{k-1,m}.
\]

For any polynomial $p = p(z)$ let $[z^k] p$ stand for the coefficient of $z^k$ in the polynomial assumed. 
For $\Phi := \sum_{k \ge 2} x_k z^k$ we have $\phi_{k,m} = [z^k] \Phi^m$.
Denote also $\Psi := \sum_{k \ge 2} x_k z^k / y_1^k$. Note that this power series is obtained from $\Phi$ by a formal rescaling of $z$ variable.

Equation above is equivalent to
\begin{align*}
    & y_\alpha - \frac{y_1^\alpha}{\alpha} = y_1^\alpha [z^\alpha] \left( \sum_{k=2}^{\alpha+1}\sum_{m=1}^{k-1} \frac{(m+\alpha-k)!}{(\alpha+1-k)! m!} z^{\alpha+1-k} \Psi^m \right)
    \\
    & \Leftrightarrow \ y_\alpha - \frac{y_1^\alpha}{\alpha} = y_1^\alpha [z^\alpha] \left( \sum_{k=2}^{\alpha} \frac{z^{\alpha+1-k}}{(\alpha+1-k)} \frac{1}{(1 - \Psi)^{\alpha+1-k}} + \sum_{m \ge 1} \frac{\Psi^m}{m} \right)
    \\
    & \Leftrightarrow \ 
    y_\alpha = y_1^\alpha [z^\alpha] \left( - \log(1 - \frac{z}{1-\Psi}) - \log(1-\Psi)\right).
\end{align*}
Rescaling formally $z$ on the both sides and collecting $y_\bullet$ into a power series the proposition follows.
\end{proof}

The following proposition shows that the dependance of $\p_1\p_j f$ on $\p_0\p_\bullet f$ is given via the Schur polynomials.
\begin{proposition}\label{prop: A ext to Toda}
    Let $f$ satisfy \eqref{eq: A extended system} then we have
    \begin{equation}
        \p_1\p_j f = \sum_{m \ge 1} \frac{(-1)^{m-1}}{m!} \sum_{\gamma_1+\dots+\gamma_m = j+1} \prod_{k=1}^m \p_0\p_{\gamma_k} f.
    \end{equation}
\end{proposition}
\begin{proof}
Assume the notation introduced in the proof of Proposition~\ref{prop: A ext expanded}.
Exponentiating the equality of Proposition~\ref{prop: A ext expanded} we get
\begin{align*}
    \Phi(z) = 1 - \exp(- \sum_{\alpha \ge 1} y_\alpha z^\alpha) - y_1z.
\end{align*}
Expanding the exponent in a power series and combining the coefficients of $z^\alpha$ on the both sides we get exactly the desired equation.
\end{proof}

Our goal now is to relate the system \eqref{eq: A KP system} and \eqref{eq: A extended system} to dispersionless mKP hierarchy.

\subsection{Dispersionless mKP hierarchy}
Consider the notation
\[
    D(z) := \sum_{k \ge 1} \frac{z^{-k}}{k} \frac{\p}{\p t_k}.
\]
For a function $f =f (t_0,t_1,t_2,\dots)$, the dispersionless limit of mKP hierarchy can be written in a Fay form as the following equality of the formal power series in $z^{-1}$, $w^{-1}$.
\begin{equation}\label{eq: mKP}
    e^{D(z)D(w) f} = \frac{z \cdot e^{-D(z)\p_0f} - w \cdot e^{-D(w)\p_0f}}{z-w},
\end{equation}
The coefficient of $w^{-1}$ on the both sides gives the following equality
\begin{equation}\label{eq: KP from mKP}
    z - D(z)\p_1 f = z e^{-D(z)\p_0 f}.
\end{equation}
Substituting it back to \eqref{eq: mKP} gives 
\begin{equation}\label{eq: KP}
    e^{D(z)D(w) f} = 1 - \frac{D(z)\p_1f - D(w)\p_1f}{z-w}.
\end{equation}
This equation is exactly the Fay form of the dispersionless limit of KP hierarchy of the function $f$ assumed as a function of $t_1,t_2,\dots$ with $t_0$ being fixed.

\begin{theorem}\label{theorem: dispersionless mKP}
    The system of equations \eqref{eq: A KP system} and \eqref{eq: A extended system} coincides with the dispersionless mKP hierarchy after the change of the variables $t_k \to t_k/k$, $k \ge 1$.
\end{theorem}
\begin{proof}
    After the change of the variables given, the system \eqref{eq: A KP system} coincides with with the Fay form of dispersionless KP by \eqref{eq: dKP flows expansion} (see also \cite{BDbN}). 
    It follows immediately from Proposition~\ref{prop: A ext to Toda} that \eqref{eq: KP from mKP} holds for a solution to \eqref{eq: A extended system}. Therefore substituting it to \eqref{eq: A KP system} on gets \eqref{eq: A extended system}.
\end{proof}

\subsection{h--deformation}
Full mKP hierarchy can be obtained from its dispersionless limit by the following procedure extending the approach of \cite{NZ}.

 Let $\tau_\hbar(s,\bt) := \tau(\hbar^{-1}s,\hbar^{-1} \bt)$. Assume it to be expanded by $\log\tau_\hbar = \sum_{g\ge 0}\hbar^{g-2}F_g$. Denote ${f := \hbar^2 \log\tau_\hbar}$.
 Consider also the operators
 \[
    \Delta(z):= \hbar^{-1}\left( e^{\hbar D(z)} - 1\right), \quad \Delta^o := \hbar^{-1}\left( e^{-\hbar \p_0} - 1\right).
 \]
 Note that $\Delta^o$ only involves the differentiations w.r.t. $t_0$ while $\Delta(z)$ only the differentiations w.r.t. $t_k$, $k \ge 1$.
 
 The $\hbar$--deformed version of \eqref{eq: mKP} above is the following equation
 \begin{equation}\label{eq: full mKP in Fay}
    e^{\Delta(z)\Delta(w) f} = \frac{z \cdot e^{-\Delta(z) \Delta^0 f} - w \cdot e^{-\Delta(w)\Delta^0f}}{z-w}.
 \end{equation}
 It coincides with the Fay form of mKP hierarchy (cf. \cite[Section 3.2]{T}).

 Consider $\p^\hbar_k$ defined by
 \begin{equation*}
    \p^\hbar_0 := \hbar^{-1}\left( e^{-\hbar \p_0} - 1\right), \quad \sum_{k \ge 1} \frac{z^{-k}}{k} \p^\hbar_k = \Delta(z)
 \end{equation*}
 Comparing the coefficients of $z^{-a}w^{-b}$ on the both sides of \eqref{eq: full mKP in Fay} one gets exactly
 \begin{align*}
    \p_\alpha^\hbar\p_\beta^\hbar f & = \sum_{m \ge 1} \sum_{\gamma_\bullet} R^{\mathrm{A}}_{\alpha,\beta; \gamma_1,\dots,\gamma_m} \prod_{k=1}^m \p_1\p_{\gamma_k}^\hbar f,
    \\
    \p_{0}^\hbar\p_\alpha^\hbar f & = \sum_{m \ge 1} \sum_{\gamma_\bullet} R^{\mathrm{A},ext}_{\alpha; \gamma_1,\dots,\gamma_m} \prod_{k=1}^m \p_1\p_{\gamma_k}^\hbar f.
\end{align*}

This procedure recovers full mKP hierarchy from the dispersionless hierarchy that we've constructed from the family of pairs $\lbrace F^c_{A_N}, F^o_{A_N} \rbrace_{N \ge 1}$.

\section{Dispersionless open--closed system of type D}
The following stabilization condition was proved in \cite{BDbN}.

For any $4 \le N < N$ we have
\begin{align*}
     &\left.\frac{\p^2 F^c_{D_N}(\bt)}{\p t_\alpha \p t_\beta}\right|_{\forall \gamma \; t_{\gamma} = s_{N-\gamma}; \ t_N = s_{-1}} = \left.\frac{\p^2 F^c_{D_M}}{\p t_\alpha \p t_\beta}\right|_{\forall \gamma \; t_{\gamma} = s_{M-\gamma}; \ t_M = s_{-1}},
         \qquad \forall \alpha+\beta < N,
     \\
     &\left.\frac{\p^2 F^c_{D_N}(\bt)}{\p t_N \p t_\beta}\right|_{\forall \gamma \; t_{\gamma} = s_{N-\gamma}; \ t_N = s_{-1}} = \left.\frac{\p^2 F^c_{D_M}(\bt)}{\p t_M \p t_\beta}\right|_{\forall \gamma \; t_{\gamma} = s_{M-\gamma}; \ t_M = s_{-1}}, \qquad \forall \beta < N,
\end{align*}
understood as an equality of polynomials in $s_\bullet$.

\begin{proposition}
    For any $4 \le N < M$ and $ 1\le \alpha < N$ we have
    \begin{align*}
      \frac{\p F_{D_N}^o}{\p t_\alpha} \mid_{t_{N-\beta} = s_\beta, \ t_{N} = s_{-1}} = \frac{\p F_{D_M}^o}{\p t_\alpha} \mid_{t_{M-\beta} = s_\beta, \ t_{M} = s_{-1}},
      \\
      \frac{\p F_{D_N}^o}{\p t_{N}} \mid_{t_{N-\beta} = s_\beta, \ t_{N} = s_{-1}} = \frac{\p F_{D_M}^o}{\p t_{M}} \mid_{t_{M-\beta} = s_\beta, \ t_{M} = s_{-1}},
    \end{align*}
    assumed as polynomials in $s_{-1},s_1,\dots,s_{M}$ and Laurent polynomial in $t_0$.
\end{proposition}
\begin{proof}
    Note that the change of the variables above does not affect the variable $t_0$ and we can consider $F^o_{D_N}$ and $F^o_{D_M}$ as Laurent polynomials in $t_0$.

    The second equality is straightforward by the explicit form of $F^o_{D_N}$.
    To show the first equality we only have to consider the summands of $F^o_{D_N}$ and $F^o_{D_M}$ involving the variable $t_0$. This follows from \eqref{eq: Dn essential coordinate via flat} and Theorem 4.9 of \cite{BDbN}.
\end{proof}

For $\alpha,\beta < N$, and $\gamma_k < N$ denote $N = \alpha+\beta -1$ and $\bar N := -1$, $\bar \gamma := N-\gamma$ for $1 \le \gamma < N$. Set
\begin{align*}
    & R^{(\mathrm{D},1)}_{\alpha,\beta; \bar\gamma_1,\dots,\bar\gamma_m} = \left.\frac{1}{m!} \frac{\p^{m+2} F^c_{D_N}}{\p_\alpha \p_\beta \p t_{\gamma_1}\cdot\dots\cdot \p t_{\gamma_m}} \right|_{\bt=0},
    \\
    & R^{(\mathrm{D},2)}_{\alpha; \bar\gamma_1,\dots,\bar\gamma_m} = \left.\frac{1}{m!} \frac{\p^{m+3} F^c_{D_N}}{\p t_N \p t_N \p t_\alpha \p t_{\gamma_1}\cdot\dots\cdot \p t_{\gamma_m}} \right|_{\bt=0},
    \\
    & R^{(\mathrm{D},ext,1)}_{\alpha; \bar\gamma_1,\dots,\bar\gamma_m} = 
    \frac{1}{m!} \left.\frac{\p^{m+1} F_{D_N}^o}{\p t_\alpha \p t_{\gamma_1}\cdots \p t_{\gamma_m}} \right|_{\bt =  0}.
\end{align*}
It was proved in \cite{BDbN} that $R^{(\mathrm{D},1)}_\bullet$ and $R^{(\mathrm{D},2)}_\bullet$ are well--defined. 
Note that despite $F^o_{D_N}$ being Laurent polynomial, according to the proposition above all coefficients $R^\bullet$ are well-defined.

The data above does not collect all the information of $F^o_{D_N}$ that is stabilized as $N$ grows. In particular we have $\p_N \p_N F^o_{D_N} = 1 / t_0$, for which $\bt = 0$ is not defined. However we still can add the corresponding flow~\eqref{eq: DN ext flow 2} to the system of PDEs we build.

For a function $f = f(t_{-1},t_0,t_1,t_2,\dots)$ consider the system of PDEs
\begin{align}
    \p_\alpha \p_\beta f &= \sum_{m \ge 1}\sum_{\gamma_1,\dots,\gamma_m} R^{(\mathrm{D},1)}_{\alpha,\beta; \gamma_1,\dots,\gamma_m} \prod_{k=1}^m \p_1 \p_{\gamma_a} f 
    \label{eq: DN flow A}
    \\
    \p_{-1} \p_\alpha f &= \p_{-1} \p_1 f \cdot \sum_{m\ge 1} \sum_{\gamma_1,\dots,\gamma_m} R^{(\mathrm{D},2)}_{\alpha; \gamma_1,\dots,\gamma_m} \prod_{k=1}^m \p_1 \p_{\gamma_a} f,
    \label{eq: DN flow B}
    \\
    \p_0 \p_\alpha f &= \sum_{m\ge 1} \sum_{\gamma_1,\dots,\gamma_m} R^{(\mathrm{D},ext)}_{\alpha; \gamma_1,\dots,\gamma_m} \prod_{k=1}^m \p_1 \p_{\gamma_a} f,
    \label{eq: DN ext flow 1}
    \\
    \p_0 \p_{-1} f & \cdot \p_1\p_0 f = \p_1\p_{-1}f,
    \label{eq: DN ext flow 2}
\end{align}
for all $\alpha,\beta \ge 2$ and $\gamma_\bullet \ge -1$.

We show that this system of equations coincides with the system given in Introduction after setting $t_{-1} = \bar t_1$.

\begin{theorem}\label{theorem: D type}~
    \begin{enumerate}
     \item The numbers $R^{(\mathrm{D},1)}_{\alpha,\beta; \gamma_1,\dots,\gamma_m}$ and $R^{(\mathrm{D},2)}_{\alpha,\beta; \gamma_1,\dots,\gamma_m}$ are series expansions of $ \p_\alpha \p_\beta F_{D_N}^c$ and $\p_N \p_\alpha F_{D_N}^c$ respectively written in the coordinates $t_{\bar 1}, \dots, t_{\bar N}$.
     \item The system \eqref{eq: DN flow A},\eqref{eq: DN flow B} and \eqref{eq: DN ext flow 1},\eqref{eq: DN ext flow 2} is compatible.
     \item 
     The function 
    \[
    \tilde f = F_{D_N}^c(t_{\overline 1},\dots,t_{\overline N}) + \int F^o_{D_N}(t_{\overline 1},\dots,t_{\overline N},t_0) dt_0
    \] 
    satisfies this system for $\alpha+\beta \le N$.
    \end{enumerate}
\end{theorem}
\begin{proof}
    Part (1) follows immediately from \eqref{eq: DN^c expression}. 
    After this in order to show (2) we may use the PDEs given in Introduction. Compatibility of \eqref{eq: DN flow A} and \eqref{eq: DN flow B} was proved in \cite{BDbN}. 
    
    Let $\kappa = \alpha+\beta-1$. In the following formulae we use $\p_\bullet\p_1\p_\alpha f = \p_1\p_\bullet\p_\alpha f$. By the chain rule we have
    \begin{align*}
        \p_0 &(\p_\alpha\p_\beta f) = \sum_{\delta = 1}^\kappa \p_0 \p_1\p_{\bar \delta} f \cdot \frac{\p^3 F^c_{D_\kappa}}{\p t_\delta \p t_\alpha \p t_\beta}\mid_{ t_{\kappa-\gamma} = \p_1 \p_\gamma f, \ t_\kappa = \p_1 \bar\p_1 f},
        \\
        & = \p_1 \frac{\p F^o_{D_\kappa}}{\p t_{\bar \delta}}\cdot \frac{\p^3 F^c_{D_\kappa}}{\p t_\delta \p t_\alpha \p t_\beta}
        = \p_1\p_1\p_0 f \cdot \frac{\p F^o_{D_\kappa}}{\p t_{\bar \delta}\p t_0} \frac{\p^3 F^c_{D_\kappa}}{\p t_\delta \p t_\alpha \p t_\beta} 
        + \sum_{\sigma = 1}^\kappa \p_1\p_1\p_\sigma f \cdot \frac{\p F^o_{D_\kappa}}{\p t_{\bar \delta}\p t_{\bar \sigma}} \frac{\p^3 F^c_{D_\kappa}}{\p t_\delta \p t_\alpha \p t_\beta},
    \end{align*}
    where we skip the variable substitution on the second line in order to simplify the formulae.
    Similarly we have
    \begin{align*}
        \p_\beta \p_0 &\p_\alpha f = \sum_{\delta = 1}^\kappa \p_\beta \p_1 \p_{\bar \delta} f \cdot  \frac{\p^2 F^o_{D_\kappa}}{\p t_\delta \p t_\alpha} \mid_{ t_{\kappa-\gamma} = \p_1 \p_\gamma f, \ t_\kappa = \p_1 \bar\p_1 f, \ t_0 = \p_1 \p_0 f},
        \\
        & \qquad\qquad + \p_\beta \p_1 \p_{0} f \cdot  \frac{\p^2 F^o_{D_\kappa}}{\p t_0 \p t_\alpha} \mid_{ t_{\kappa-\gamma} = \p_1 \p_\gamma f, \ t_\kappa = \p_1 \bar\p_1 f, \ t_0 = \p_1 \p_0 f},
        \\
        &= \sum_{\delta = 1}^\kappa \p_1 \frac{\p^2 F^c_{D_\kappa}}{\p t_\beta \p t_{\bar \delta}} \cdot  \frac{\p^2 F^o_{D_\kappa}}{\p t_\delta \p t_\alpha}
        + \p_1 \frac{\p F^o_{D_\kappa}}{\p t_\beta} \cdot  \frac{\p^2 F^o_{D_\kappa}}{\p t_0 \p t_\alpha}
        \\
        &= \sum_{\sigma=1}^\kappa \p_1\p_1 \p_\sigma f \cdot  \sum_{\delta = 1}^\kappa  \left( \frac{\p^3 F^c_{D_\kappa}}{\p t_\beta \p t_{\bar \delta} \p_{\bar \sigma}}  \frac{\p^2 F^o_{D_\kappa}}{\p t_\delta \p t_\alpha}
        + \frac{\p^2 F^o_{D_\kappa}}{\p t_\beta \p t_{\bar\sigma}} \frac{\p^2 F^o_{D_\kappa}}{\p t_0 \p t_\alpha} \right)
        +\p_1\p_1\p_0 f \cdot \frac{\p^2 F^o_{D_\kappa}}{\p t_\beta \p t_0} \cdot  \frac{\p^2 F^o_{D_\kappa}}{\p t_0 \p t_\alpha}.
    \end{align*}
    The coefficient of $\p_1\p_1\p_0 f$ on the both sides is the same if and only if
    \begin{align*}
     \sum_{\delta=1}^\kappa\frac{\p F^o_{D_\kappa}}{\p t_0 \p t_{\bar \delta}} \frac{\p^3 F^c_{D_\kappa}}{\p t_\delta \p t_\alpha \p t_\beta} = \frac{\p^2 F^o_{D_\kappa}}{\p t_\beta \p t_0}  \frac{\p^2 F^o_{D_\kappa}}{\p t_0 \p t_\alpha}.
    \end{align*}
    The metric $\eta$ defined by $F^c_{D_\kappa}$ has only the following non--zero entries $\eta^{\kappa\kappa} = 1$ and $\eta^{\alpha\beta} = \delta^{\alpha+\beta,N}$. With our choice of $\bar\beta$ this is equivalent to $\eta^{\alpha\beta} = \delta^{\alpha,\bar\beta}$.
    
    By open WDVV equation the difference of LHS and RHS is equal to $\frac{\p^2 F^o_{D_\kappa}}{\p t_\alpha \p t_\beta} \frac{\p^2 F^o_{D_\kappa}}{\p t_0 \p t_0}$ that vanishes for $\alpha+\beta \le \kappa$ and $\alpha,\beta \neq \kappa$ due to \eqref{eq: D open potential quasihomogeneity}. 
    For any $\sigma \neq 0$ the coefficients of $\p_1\p_1\p_\sigma f$ of the expressions above coincide by the same reasoning.

    Consider now part (3).
    Because of the special dependance of $F^c_N$ and $F^o_N$ on the variable $t_1$ we have
    \[
        \p_1\p_{\overline\gamma} \tilde f = \begin{cases}
                           \p_1\p_{\overline\gamma} F^c_N = t_{\gamma}, \quad &\text{ if } \quad \overline\gamma > 0,
                           \\
                           \p_1\p_{\overline\gamma} F^c_N = t_N, \quad &\text{ if } \quad \overline\gamma = -1,
                           \\
                           \p_1 F^o_N = t_0 , \quad &\text{ if } \quad \overline\gamma = 0.
                          \end{cases}
    \]
    To show (3) let's show that applied to $\tilde f$ assumed both \eqref{eq: A KP system} and \eqref{eq: A extended system} just provide the series expansions of $\p_\alpha\p_\beta F^c_N$ and $\p_\alpha F^o_N$.
    
    It follows from \eqref{eq: A open potential quasihomogeneity} that $\p_\alpha \p_\beta F^o_{A_N} \equiv 0$ whenever $\alpha+\beta \le N+1$.
    The function $F^c_N$ does not depend on $t_0$ and therefore \eqref{eq: A KP system} holds for $\tilde f$.
    
    Similarly $\p_0\p_\alpha \tilde f = \p_\alpha F^o_N$ what approves \eqref{eq: A extended system} for $\tilde f$.
\end{proof}


\end{document}